\documentclass[twocolumn]{revtex4}

\usepackage{amsmath,amsfonts,amssymb}

\newlength{\halfpagewidth}

\divide\halfpagewidth by 2

\newtheorem{definition}{Definition}
\newtheorem{proposition}[definition]{Proposition}
\newtheorem{Lemma}{Lemma}

\newtheorem{Theorem}[definition]{Theorem}

\newtheorem{conjecture}[definition]{Conjecture}

\newtheorem{remark}[definition]{Remark}
\newtheorem{example}[definition]{Example}
\newtheorem{question}[definition]{Question}

\def\squareforqed{\hbox{\rlap{$\sqcap$}$\sqcup$}}
\def\qed{\ifmmode\squareforqed\else{\unskip\nobreak\hfil
		\penalty50\hskip1em\null\nobreak\hfil\squareforqed
		\parfillskip=0pt\finalhyphendemerits=0\endgraf}\fi}
\def\endenv{\ifmmode\;\else{\unskip\nobreak\hfil
		\penalty50\hskip1em\null\nobreak\hfil\;
		\parfillskip=0pt\finalhyphendemerits=0\endgraf}\fi}
\newenvironment{proof}{\noindent \textbf{{Proof.~} }}{\qed}
\def\Dbar{\leavevmode\lower.6ex\hbox to 0pt
	{\hskip-.23ex\accent"16\hss}D}
\makeatletter
\def\url@leostyle{%
	\@ifundefined{selectfont}{\def\UrlFont{\sf}}{\def\UrlFont{\small\ttfamily}}}
\makeatother
\def\bcj{\begin{conjecture}}
	\def\ecj{\end{conjecture}}
\def\bcr{\begin{corollary}}
	\def\ecr{\end{corollary}}
\def\bd{\begin{definition}}
	\def\ed{\end{definition}}
\def\bea{\begin{eqnarray}}
\def\eea{\end{eqnarray}}
\def\bem{\begin{enumerate}}
	\def\eem{\end{enumerate}}
\def\bex{\begin{example}}
	\def\eex{\end{example}}
\def\bim{\begin{itemize}}
	\def\eim{\end{itemize}}
\def\bl{\begin{lemma}}
	\def\el{\end{lemma}}
\def\bma{\begin{bmatrix}}
	\def\ema{\end{bmatrix}}
\def\bpf{\begin{proof}}
	\def\epf{\end{proof}}
\def\bpp{\begin{proposition}}
	\def\epp{\end{proposition}}
\def\bqu{\begin{question}}
	\def\equ{\end{question}}
\def\br{\begin{remark}}
	\def\er{\end{remark}}
\def\bt{\begin{theorem}}
	\def\et{\end{theorem}}

\def\btb{\begin{tabular}}
	\def\etb{\end{tabular}}

\newcommand{\nc}{\newcommand}

\def\r{\rho}
\def\s{\sigma}

\def\c{\chi}

\nc{\bbA}{\mathbb{A}} \nc{\bbB}{\mathbb{B}} \nc{\bbC}{\mathbb{C}}
\nc{\bbD}{\mathbb{D}} \nc{\bbE}{\mathbb{E}} \nc{\bbF}{\mathbb{F}}
\nc{\bbG}{\mathbb{G}} \nc{\bbH}{\mathbb{H}} \nc{\bbI}{\mathbb{I}}
\nc{\bbJ}{\mathbb{J}} \nc{\bbK}{\mathbb{K}} \nc{\bbL}{\mathbb{L}}
\nc{\bbM}{\mathbb{M}} \nc{\bbN}{\mathbb{N}} \nc{\bbO}{\mathbb{O}}
\nc{\bbP}{\mathbb{P}} \nc{\bbQ}{\mathbb{Q}} \nc{\bbR}{\mathbb{R}}
\nc{\bbS}{\mathbb{S}} \nc{\bbT}{\mathbb{T}} \nc{\bbU}{\mathbb{U}}
\nc{\bbV}{\mathbb{V}} \nc{\bbW}{\mathbb{W}} \nc{\bbX}{\mathbb{X}}
\nc{\bbZ}{\mathbb{Z}}

\nc{\bA}{{\bf A}} \nc{\bB}{{\bf B}} \nc{\bC}{{\bf C}}
\nc{\bD}{{\bf D}} \nc{\bE}{{\bf E}} \nc{\bF}{{\bf F}}
\nc{\bG}{{\bf G}} \nc{\bH}{{\bf H}} \nc{\bI}{{\bf I}}
\nc{\bJ}{{\bf J}} \nc{\bK}{{\bf K}} \nc{\bL}{{\bf L}}
\nc{\bM}{{\bf M}} \nc{\bN}{{\bf N}} \nc{\bO}{{\bf O}}
\nc{\bP}{{\bf P}} \nc{\bQ}{{\bf Q}} \nc{\bR}{{\bf R}}
\nc{\bS}{{\bf S}} \nc{\bT}{{\bf T}} \nc{\bU}{{\bf U}}
\nc{\bV}{{\bf V}} \nc{\bW}{{\bf W}} \nc{\bX}{{\bf X}}
\nc{\bZ}{{\bf Z}}

\nc{\cA}{{\cal A}} \nc{\cB}{{\cal B}} \nc{\cC}{{\cal C}}
\nc{\cD}{{\cal D}} \nc{\cE}{{\cal E}} \nc{\cF}{{\cal F}}
\nc{\cG}{{\cal G}} \nc{\cH}{{\cal H}} \nc{\cI}{{\cal I}}
\nc{\cJ}{{\cal J}} \nc{\cK}{{\cal K}} \nc{\cL}{{\cal L}}
\nc{\cM}{{\cal M}} \nc{\cN}{{\cal N}} \nc{\cO}{{\cal O}}
\nc{\cP}{{\cal P}} \nc{\cQ}{{\cal Q}} \nc{\cR}{{\cal R}}
\nc{\cS}{{\cal S}} \nc{\cT}{{\cal T}} \nc{\cU}{{\cal U}}
\nc{\cV}{{\cal V}} \nc{\cW}{{\cal W}} \nc{\cX}{{\cal X}}
\nc{\cZ}{{\cal Z}}

\nc{\hA}{{\hat{A}}} \nc{\hB}{{\hat{B}}} \nc{\hC}{{\hat{C}}}
\nc{\hD}{{\hat{D}}} \nc{\hE}{{\hat{E}}} \nc{\hF}{{\hat{F}}}
\nc{\hG}{{\hat{G}}} \nc{\hH}{{\hat{H}}} \nc{\hI}{{\hat{I}}}
\nc{\hJ}{{\hat{J}}} \nc{\hK}{{\hat{K}}} \nc{\hL}{{\hat{L}}}
\nc{\hM}{{\hat{M}}} \nc{\hN}{{\hat{N}}} \nc{\hO}{{\hat{O}}}
\nc{\hP}{{\hat{P}}} \nc{\hR}{{\hat{R}}} \nc{\hS}{{\hat{S}}}
\nc{\hT}{{\hat{T}}} \nc{\hU}{{\hat{U}}} \nc{\hV}{{\hat{V}}}
\nc{\hW}{{\hat{W}}} \nc{\hX}{{\hat{X}}} \nc{\hZ}{{\hat{Z}}}

\nc{\hn}{{\hat{n}}}

\def\max{\mathop{\rm max}}
\def\min{\mathop{\rm min}}

\def\rank{\mathop{\rm rank}}

\newcommand{\bra}[1]{\langle#1|}
\newcommand{\ket}[1]{|#1\rangle}

\begin{document}
	\title{A coherence quantifier based on the quantum optimal transport cost}
	
	\author{Xian Shi}\email[]
	{shixian01@gmail.com}
\affiliation{College of Information Science and Technology, Beijing University of Chemical Technology, Beijing 100029, China}
	
	%
	
	
	
	\date{\today}
	

\begin{abstract}
\indent In this manuscript, we present a coherence measure based on the quantum optimal transport cost in terms of convex roof extended method. We also obtain the analytical solutions of the quantifier for pure states. At last, we propose an operational interpretation of the coherence measure for pure states.
\end{abstract}

\maketitle
\section{introduction}
\indent Quantum coherence is one of the crucial resources in quantum information processing \cite{streltsov2017}. It springs from the state superposition principle, which distinguishes quantum information from classical high. Coherence has high relations with other quantum resources, such as quantum entanglement \cite{horodecki2009quantum}, quantum nonlocality \cite{brunner2014bell}, and so on. Coherence also plays crucial roles in quantum algorithms \cite{shi2017coherence,rastegin2018role,ahnefeld2022coherence,naseri2022entanglement,ye2023tsallis,li2023evolution}, thermodynamical systems \cite{aaberg2014catalytic,narasimhachar2015low,lostaglio2015description}, transport theory \cite{witt2013stationary}, and biology \cite{plenio2008dephasing,huelga2013vibrations}. \par

One of the essential problems in quantum information theory is how to quantify coherence. In 2014, the authors proposed a framework with four postulates for a proper coherence measure \cite{baumgratz2014quantifying}, there the authors also presented proper coherence measures, $\mathrm{l}_1$ norm of coherence and the relative entropy of coherence. In 2015, Du $et$ $al.$ proposed a method to construct coherence measures and considered the optimal conversion for coherent states \cite{du2015coherence}. In 2016, Napoli $et$ $al.$ proposed the robustness of coherence and provided an operational interpretation of the coherence measure \cite{napoli2016robustness}. Subsequently, some coherence measures are offered based on the trace norm distance \cite{rana2016trace}, Tsallis-$q$ entropy \cite{rastegin2016}, Fisher information \cite{li2021quantum,yu2022quantifying}, and so on.  In addition, Yu $et$ $al.$ proposed an alternative framework for quantifying coherence, and they also showed that the coherence measure in terms of $1$-norm is improper based on the framework they proposed here \cite{yu2022quantifying}. Next, due to the relationship between coherence and entanglement, the approaches to building coherence measures based on entanglement monotone are shown \cite{streltsov2015measuring,zhu2017operational,tan2018coherence,Kim2023partial}. Subsequently, the methods to construct proper coherence measures based on other quantum correlations are proposed \cite{ma2016converting,Xiong2019partial}. Recently, the authors proposed a way to quantify the coherence of a quantum system via Kirkwood-Dirac quasiprobability \cite{budiyono2023quantifying}. \par 
The study of various distances between the states in quantum systems is important. Recently, the quantum Wasserstein distance, which originated from the quantum optimal transport, attracted much attention as it has been shown helpful in quantum theory. In 2021, Palma $et$ $al.$ proposed the Wasserstein distance of order 1 for $n$-qudit systems \cite{de2021quantum}. In 2022, Friedland $et$ $al.$ considered a quantum version of the Monge-Kantorovich optimal transport problem. There the authors conjected that 2-Wasserstein distance is monotone under partial traces \cite{friedland2022quantum}. Nevertheless, the author proved the conjecture is invalid and proposed a revised version of the quantum optimal transport cost \cite{muller2022monotonicity}. The authors in \cite{bistron2023monotonicity} also considered the problem of the monotonicity of a quantum 2-Wasserstein distance. Then Palma and Trevisan proposed a new quantum generalization of the Wasserstein distance, one-to-one correspondence with quantum channels \cite{palma2021quantum}. Besides the applications on the quantum states, the authors also showed that distance is helpful for quantum algorithms \cite{chakrabarti2019quantum,kiani2022learning,de2023limitations}. Recently, Toth considered how to define the quantum Wasserstein distance by optimizing over bipartite separable states \cite{toth2023quantum}. Then, is it possible to quantify the coherence by the quantum Wasserstein distance and whether a quantum task is related to the quantifier closely?

Quantum speed limits are used to denote the lower bound on the minimal time required to evolve from an initial state to a target state \cite{deffner2017quantum, caglioti2020quantum}. Recently, researchers paid much attention to the relationship between the quantum speed limit and quantum resource theories \cite{marvian2016quantum,rudnicki2021quantum,seddon2021quantifying,campaioli2022resource,mohan2022quantum,muthuganesan2022affinity,maleki2023speed}. One of the most famous examples is the Mandelstan-Tamm bound \cite{mandelstam1945quantum,caglioti2020quantum} for pairs of orthogonal pure states, and the bound can be generalized to arbitrary pairs of pure states $\ket{\phi_1}$ and $\ket{\phi_2},$ the minimal time required $\tau$ to evolve from $\ket{\phi_1}$ into $\ket{\phi_2}$ is
\begin{align}
\tau\ge \hbar\frac{\arccos(|\bra{\phi_1}\phi_2\rangle|)}{\overline{\triangle H}},\label{qst}
\end{align}
here $\overline{\triangle H}=\frac{1}{\tau}\int_0^{\tau}dt\sqrt{\bra{\phi_t}H_t^2\ket{\phi_t}-|\bra{\phi_t}H_t\ket{\phi_t}|^2}$, $H_t$ is Hamiltonian of the unitary evolution \cite{campaioli2020tightening}. 

In this paper, due to the beautiful properties of the revised quantum optimal transport cost in \cite{muller2022monotonicity}, we present a quantifier for the coherence of the quantum systems based on the cost. We also deliver the analytical solutions of the quantifier for pure states of qudit systems. Besides, we showed the quantifier is not a coherence measure. To amend the flaw, we proposed a coherence measure for mixed states based on the quantifier under the convex roof extended method. At last, we show that this measure links to the quantum speed limits for pure states.\par 
The paper is organized as follows, in Sec. \uppercase\expandafter{\romannumeral2} , we present the preliminary knowledge on the revised quantum optimal cost needed, in Sec. \uppercase\expandafter{\romannumeral3} , we present our main results. First, we consider a coherence quantifier based on the revised quantum optimal cost and show some properties of the measure. We also obtain the analytical values of the quantity for pure states of qudit system. Besides, we propose a coherence measure based on the quantum optimal cost under the convex roof extended method. At last, we obtain an operational interpretation of the coherence measure for pure states. In Sec. \uppercase\expandafter{\romannumeral4} , we end with a summary.
\section{Preliminary Knowledge}

\indent In this article, we assume $\mathcal{H}$ is a Hilbert space with finite dimension, let $\mathbb{H}(\mathcal{H}_d)=\{M|M=M^{\dagger}\}$ and $\mathcal{D}(\mathcal{H}_d)=\{\r|\r\ge0,tr\r=1\}.$ A pure state of $\mathcal{H}_d$ is a vector $\ket{\psi}\in \mathcal{H_d}$ with norm 1.\par 

In the section, we first recall the postulates for quantum coherence measures of states on finite dimensional systems, then we review the knowledge needed on quantum transport cost.

\subsection{Postulates for quantum coherence measures}
Let $\mathcal{H}$ be a $d$-dimensional Hilbert space. Let $\mathcal{E}=\{\ket{i}\}$ be the set consisting of a prescribed orthonormal basis in $\mathcal{H},$ the set of incoherent states $I_{\mathcal{E}}$ is composed of all the states that are diagonal with respect to the basis $\{\ket{i}\}$, that is, any incoherent state $\delta\in I_{\mathcal{E}}$ can be written as 
\begin{align*}
\delta=\sum_i \lambda\ket{i}\bra{i}.
\end{align*}

An important problem in coherence theory is how to quantify coherence of a state. In \cite{baumgratz2014quantifying}, the authors proposed the following postulates for functions $C:\mathcal{D}(\mathcal{H})\rightarrow \mathbb{R}^{+}$:\\
(B1). $C(\rho)\ge0$ for any $\rho\in\mathcal{D}(\mathcal{H})$, and $C(\rho)=0$ if and only if $\rho\in \mathcal{I}_{\mathcal{E}}$;\\
(B2). Monotonicity under incoherent completely positive and
trace preserving maps (ICPTP) $\Psi$, $C(\r)\ge C(\Psi(\r))$; \\
(B3). Monotonicity for average coherence under subselection based on measurements outcomes: $C(\rho)\ge \sum_i p_i C(\rho_i),$ here $\rho_i=\frac{K_i\rho K_i^{\dagger}}{p_i},$ $p_i=tr(K_i\rho K_i^{\dagger})$ for all $K_i$ satisfying $\sum_i K_i^{\dagger}K_i=I,$ and $K_i\mathcal{I}_{\mathcal{E}}K_i^{\dagger}\subset \mathcal{I}_{\mathcal{E}}$; \\
(B4). Non-increasing under mixing of quantum states, $\sum_i p_i C(\rho_i)\ge C(\sum_i p_i \r_i)$ for any $\{p_i,\rho_i\}.$ \par 
The Conditions (B1) and (B2) are the basic requirements for a quantity to be a coherence quantifier \cite{streltsov2017}. If $C$ satisfies all four properties, it is a faithful coherence measure.

Later, Yu $et$ $al.$ proposed an alternative framework for quantifying coherence \cite{yu2016alternative}. There (B3) and (B4) are replaced by the following,\\
 (C3). $C(p_1\r_1\oplus p_2\r_2)=p_1C(\r_1)+p_2C(\r_2)$ for block-diagonal states $\r$ in the incoherent basis. That is, when a quantity satisfies (B1), (B2) and (C3), it also satisfies (B1)-(B4).\\ 
\indent Geometric measure of coherence is a common used coherence measure based on the fidelity, which is defined as \cite{streltsov2015measuring}
\begin{align}
C_g(\rho)=1-\max_{\delta\in\mathcal{I}_{\mathcal{E}}}F(\rho,\delta),\label{geo}
\end{align}
where $F(\rho,\delta)=[tr(\sqrt{\rho}\sigma\sqrt{\rho})^{\frac{1}{2}}]^2$.  Moreover, when $\rho=\ket{\psi}\bra{\psi}$ is a pure state, then 
\begin{align}
C_g(\ket{\psi})=1-\max_i\{\psi_{ii}\},\label{g}
\end{align}
where $\psi_{ii}$ is the diagonal elements of $\psi$ with respect to the basis $\mathcal{I}_{\mathcal{E}}$.\par 

\subsection{Quantum transport cost}
\indent In this section, we first recall the definition of quantum coupling, then we review the knowledge on the quantum optimal transort cost.
\begin{definition}
	Assume $\rho_A$ and $\r_B$ are quantum states on $\mathcal{H}_A$ and $\mathcal{H}_B,$ respectively. A state $\r_{AB}$ on $\mathcal{H}_{AB}$ is called a coupling matrix for $(\r_A,\r_B)$ if 
	\begin{align*}
	tr_A\rho_{AB}=\rho_B,\hspace{3mm}tr_B\rho_{AB}=\rho_A.
	\end{align*}
\end{definition}

In this paper, we will denote the set of all quantum couplings of ($\rho_A$, $\rho_B$) as $\mathcal{B}(\rho_A,\rho_B)$. For any two states $\rho_A$ and $\rho_B,$ $\mathcal{B}(\rho_A,\rho_B)$ cannot be empty, it at least contains $\rho_A\otimes\rho_B.$\par
Next assume a bipartite Hilbert space $\mathcal{H}_{AB}=\mathcal{H}_A\otimes\mathcal{H}_B$, $\mathcal{H}_A$ and $\mathcal{H}_B$ are two Hilbert spaces with the same dimension $d$. Let $I_{\mathcal{H}_{AB}}$ and $S=\sum_{i,j=0}^{d-1}\ket{ij}\bra{ij}$ be the identity and SWAP operator of the space $\mathcal{H}_{AB},$ respectively. The symmetric subspace between $\mathcal{H}_A$ and $\mathcal{H}_B$ is defined as
\begin{align}
\mathcal{H}_A\vee\mathcal{H}_B=\{\ket{\psi}\in \mathcal{H}_{AB}|\hspace{3mm}S\ket{\psi}=\ket{\psi}\},
\end{align}
and the antisymmetric subspace between $\mathcal{H}_A$ and $\mathcal{H}_B$ is defined as
\begin{align}
\mathcal{H}_A\wedge\mathcal{H}_B=\{\ket{\psi}\in \mathcal{H}_{AB}|\hspace{3mm} S\ket{\psi}=-\ket{\psi}\}.
\end{align}

 Here we denote $P_s=\frac{1}{2}(I+S)$ and $P_a=\frac{1}{2}(I-S)$ as the projections onto the symmetric subspace and the antisymmetric subspace, respectively.	Then the quantum optimal transport cost for the states $\rho_A\in\mathcal{H}_A$,$\rho_B\in\mathcal{H}_B$ is given by 
	\begin{align}
		T(\rho,\sigma)=\min_{\c_{AB}\in \mathcal{B}(\rho_A,\rho_B)}tr[\c_{AB}P_{a}],
	\end{align}
	where the minimum takes over all the elements in $\mathcal{B}(\rho_A,\rho_B),$ which was studied in \cite{friedland2022quantum}. There the authors also showed that the cost satisfies the following monotonicity when $d=2,$ 
	\begin{align*}
	T(\Psi(\rho),\Psi(\sigma))\le T(\rho,\sigma),
	\end{align*}
	here $\rho$ and $\sigma$ take over all the states in $\mathcal{D}(\mathcal{H}_d)$ and all the channels $\Psi:\mathcal{H}_d\rightarrow \mathcal{H}_d$ when $d=2.$ However, the above inequality may be invalid when $d$ is bigger than 2 \cite{muller2022monotonicity}. In \cite{muller2022monotonicity}, the author proposed a revised cost $T_s(\rho_A,\sigma_B)$ of states $\rho_A$ and $\sigma_B$, which is is defined as
	\begin{align*}
	T_s(\rho_A,\sigma_B)=\inf_{\gamma} T(\rho\otimes\gamma,\sigma\otimes\gamma),
	\end{align*}
	where the infimum takes over all quantum states $\gamma \in \mathcal{D}(\mathcal{H}_d)$ of any dimenision $d.$ There the author showed that 
	\begin{align}
	T_s(\rho,\sigma)=&T(\rho\otimes\frac{I_2}{2},\sigma\otimes\frac{I_2}{2})\nonumber\\
	=&\min_{X_{AB},Y_{AB}} [tr(X_{AB}P_s(d)+Y_{AB}P_a(d))],\nonumber\\
\emph{s.t.}\hspace{5mm}& X_A+Y_A=\rho, \hspace{3mm} X_B+Y_B=\sigma_B.\nonumber\\
\hspace{3mm}&	X_{AB},Y_{AB}\ge0.\label{s1}
	\end{align}
 This quantity satisfies the unitary invariant and monotonicity. In the Appendix \ref{app}, we present the dual program of $T_s(\r,\s)$.

\indent 

\section{Main Results}
\indent In this section, we first propose a coherence quantifier based on the revised cost (\ref{s1}). 
Then we obtain the analytical formulas of pure states in terms of this coherence quantifier. Besides we propose a coherence measure based on the revised cost. At last, we present a relation between the measure and quantum speed limits for pure states.

 Assume $\rho\in \mathcal{D}(\mathcal{H}),$ $\tilde{T}(\rho)$ is defined as
\begin{align*}
\tilde{T}(\rho)=\min_{\delta\in\mathcal{I}_{\mathcal{E}}}T_s(\rho,\delta),
\end{align*}
where the minimum takes over all the states in $\mathcal{I}_{\mathcal{E}}.$\\

\begin{Theorem}
For a finite dimensional Hilbert space $\mathcal{H},$ $\tilde{T}(\cdot)$ satisfies (B1), (B2) and (B4). Moreover, it satisfies subadditivity for product states.
\end{Theorem}
\begin{proof}
	For the property (B1), when $\rho$ is a state, due to the definition of $\tilde{T}_s,$ $\tilde{T}_s(\rho)\ge0.$ When $\rho=\sum_ip_i\ket{i}\bra{i}\in \mathcal{I}_{\mathcal{E}},$ then we could choose $\delta=\r,$ and $X_{AB}=0,Y_{AB}=\sum_i \sqrt{p_ip_j}\ket{ii}\bra{jj},$ 
	\begin{align*}
	\tilde{T}(\r)\le&
	tr[X_{AB}P_s(d)+Y_{AB}P_a(d)]\\
	=&\sum_{i,j}\sqrt{p_ip_j}\bra{ii}P_a(d)\ket{jj}\\
	=&0.
	\end{align*}
	As $T_{s}(\r,\s)$ is nonnegative, then $\tilde{T}(\rho)=0.$

	 Next we show $\tilde{T}(\cdot)$ is faithful. Let $\rho$ be a mixed state on the Hilbert space $\mathcal{H}$ with $\mathcal{T}_s(\rho)=0$, next let $\s$ be the optimal incoherent state for $\r$ in terms of $\tilde{T}(\cdot)$. Let $X$ and $Y$ be the optimal in terms of $\tilde{T}$ for the couple $(\rho,\sigma),$ as $X,Y\ge0$ and $tr(XP_s(d)+YP_a(d))=0$, then $tr(XP_s(d))=tr(YP_a(d))=0.$ Let $$A=\{\ket{\phi_i}|\ket{\phi_i}=\ket{ii},i=0,1,\cdots,d-1\},$$
	$$B=\{\ket{\phi^{-}_{ij}}|\ket{\phi_{ij}^{-}}=\frac{\ket{ij}-\ket{ji}}{\sqrt{2}},0\le i<j\le d-1.\},$$
	$$C=\{\ket{\phi_{ij}^{+}}|\ket{\phi_{ij}^{+}}=\frac{\ket{ij}+\ket{ji}}{\sqrt{2}},0\le i<j\le d-1\},$$
	then the sets $A\cup B\cup C$, $A\cup C$ and $B$ constitute the orthonormal base of the space $\mathcal{H}\otimes\mathcal{H}$, $Range(P_s(d)),$ and $Range(P_a(d)),$ respectively. As $tr(XP_s(d))=0,$ the state $\ket{\psi}$ is in the range of $X$ is in the asymmetric space, that is, 
	\begin{align*}
	X_{AB}=&\sum_{0\le i<j\le d-1}\sum_{0\le k<l\le d-1}\mu_{ijkl}\ket{\phi_{ij}^{-}}\bra{\phi_{kl}^{-}},\\ 
	X_A=&X_B.
	\end{align*}
	Similarly, $Y_A=Y_B.$ At last, due to the definition (\ref{s1}) of $T_s(\rho,\sigma)$, $\rho=\sigma$ is incoherent. We finish the proof of the faithfulness of $\tilde{T}(\cdot).$

	For the property (B2). In \cite{muller2022monotonicity}, the author showed that for any pair of quantum states $\rho,\s\in\mathcal{D}(\mathcal{H}),$ and any quantum channel $\Phi,$ 
	\begin{align}
	T_s(\Phi(\r),\Phi(\s))\le T_s(\r,\s). \label{m}
	\end{align}
	Assume $\Psi$ is any ICPTP, $\r$ is any quantum state, $\s$ is the optimal incoherent state in terms of $T_s(\r)$,
	\begin{align*}
	\tilde{T}(\rho)=&T_s(\r,\s)\\
	\ge& T(\Psi(\r),\Psi(\s))\\
	\ge&\tilde{T}(\Psi(\r)),
	\end{align*}
	here the first inequality is due to (\ref{m}). As $\s$ is an incoherent state, and $\Psi$ is ICPTP, $\Psi(\s)$ is incoherent, then combing the definition of $\tilde{T}(\Psi(\r)),$ we finish the proof of (B2).

	For the property $(B4)$ and subadditivity for product states, we place their proof in Theorem \ref{tcon} and Theorem \ref{subpro} of the Appendix \ref{app}, respectively . 
\end{proof}

\indent Next we show the analytical solutions of $\tilde{T}(\rho)$ when $\rho=\ket{\phi}\bra{\phi}$ is a pure state.
\begin{Theorem}\label{ps}
	Assume $\ket{\phi}=\sum_{i=0}^{d-1}{\lambda_i}\ket{i}$ is a pure state, $|\lambda_0|\ge|\lambda_1|\ge\cdots\ge|\lambda_{d-1}|\ge0$ and $\sum_{i=0}^{d-1}|\lambda_i|^2=1,$ then
	\begin{align}
	\tilde{T}(\ket{\phi})=\frac{1-|\lambda_0|^2}{2}.
	\end{align}
	Moreover, it owns the following relationship with the geometric measure of coherence
	\begin{align}
\tilde{T}(\ket{\phi})=\frac{C_g(\ket{\phi})}{2}
	\end{align}
\end{Theorem}
\begin{proof} 
	Let $\sigma=\sum_kq_k\ket{k}\bra{k}$ be the optimal incoherent state for $\rho$ in terms of $\tilde{T}_s(\cdot)$. Let $X_{AB}$ and $Y_{AB}$ be the optimal in terms of (\ref{s1}) for the couple $(\rho,\sigma),$ that is, 
\begin{align*}
T_s(\rho,\sigma)=&\min_{\sigma\in \mathcal{I}_{\mathcal{E}}}[tr(X_{AB}P_s(d)+Y_{AB}P_a(d))]\\
\emph{s.t.}\hspace{5mm}& X_A+Y_A=\rho, \hspace{3mm} X_B+Y_B=\sigma\\
&X_{AB},Y_{AB}\ge 0.
\end{align*}
As $\rho$ is a pure state, $X_A$ and $Y_A$ are semipositive definite with ranks 1, $X_A=m\ket{\phi}\bra{\phi},$ $Y_A=n\ket{\phi}\bra{\phi},$ $m,n\ge0$ and $m+n=1.$ Next due to Lemma \ref{rank}, $X_{AB}=m\ket{\phi}\bra{\phi}\otimes X_B$, $Y_{AB}=n\ket{\phi}\bra{\phi}\otimes Y_B,$ then 
\begin{align}
&tr(X_{AB}P_s(d)+Y_{AB}P_a(d))\nonumber\\
=&\frac{tr[X_{AB}+Y_{AB}+F(X_{AB}-Y_{AB})]}{2}\nonumber\\
=&\frac{tr(1+F(X_{AB}-Y_{AB}))}{2}\nonumber\\
=&\frac{tr(1+F(\ket{\phi}\bra{\phi}\otimes(mX_B-nY_B)))}{2},\nonumber\\
=&\frac{1+m\bra{\phi}X_B\ket{\phi}-n\bra{\phi}Y_B\ket{\phi}}{2},\nonumber
\end{align}
then 
\begin{align}
\tilde{T}(\ket{\phi})=&\min[\frac{1+m\bra{\phi}X_B\ket{\phi}-n\bra{\phi}Y_B\ket{\phi}}{2}],\nonumber\\
=&\frac{1-|\lambda_0|^2}{2}.\label{p3}
\end{align}
Due to that $X_B,Y_B\ge 0$ and $m,n\ge 0$, we choose $X_B=0,m=0,$ $n=1$ and $Y_B=\sigma_B=\ket{0}\bra{0}$ in the second equality.

\indent Next we have the following relations between $\tilde{T}(\ket{\phi})$ and the geometric measure of coherence (\ref{geo}), 
\begin{align}
\tilde{T}(\ket{\phi})=&\frac{1-|\lambda_0|^2}{2}\nonumber\\
=&\frac{C_g(\ket{\phi})}{2},
\end{align}
the second equality is due to the formula $(\ref{g})$, that is, $C_g(\ket{\phi})=2\tilde{T}(\ket{\phi}).$
\end{proof}

Then we show that $\tilde{T}(\cdot)$ does not satisfy the property (B3) by showing the propery (C3) is invalid for $\tilde{T}(\cdot).$ Assume $\rho_1=\frac{1}{2}(\ket{0}+\ket{1})(\bra{0}+\bra{1}),$ $\rho_2=\frac{1}{3}(\ket{2}+\ket{3}+\ket{4})(\bra{2}+\bra{3}+\bra{4}),$ and $\rho=\frac{1}{2}\rho_1\oplus\frac{1}{2}\rho_2.$ By Theorem \ref{ps}, we have $\tilde{T}(\rho_1)=\frac{1}{4},$ $\tilde{T}(\r_2)=\frac{1}{3},$ $\frac{1}{2}\tilde{T}(\rho_1)+\frac{1}{2}\tilde{T}(\rho_2)=\frac{7}{24},$ next when we take
\begin{align*}
\delta=&\frac{1}{2}\ket{0}\bra{0}+\frac{1}{2}\ket{1}\bra{1}\\ 
X_{AB}=&\frac{1}{2}\ket{\psi_2}\bra{\psi_2}\otimes\ket{\psi_1^{-}}\bra{\psi_1^{-}},\\
Y_{AB}=&\frac{1}{2}\ket{\psi_1^{+}}\bra{\psi_1^{+}}\otimes\ket{\psi_1^{+}}\bra{\psi_1^{+}}, \\
\ket{\psi_1^{-}}=&\frac{1}{\sqrt{2}}(\ket{0}-\ket{1}),\\
\ket{\psi_1^{+}}=&\frac{1}{\sqrt{2}}(\ket{0}+\ket{1}),\\
\ket{\psi_2}=&\frac{1}{\sqrt{3}}(\ket{2}+\ket{3}+\ket{4})
\end{align*} 
\begin{align*}
\tilde{T}(\rho)\le& T_s(\r,\delta)\nonumber\\
\le&\frac{1+\mathrm{tr}(S(X-Y))}{2}\\
=&\frac{1}{4}<\frac{7}{24},
\end{align*}
Hence $\tilde{T}(\rho)\ne \frac{1}{2}\tilde{T}(\rho_1)+\frac{1}{2}\tilde{T}(\rho_2)$, $\tilde{T}(\cdot)$ does not satisfy $(B3).$ 

To remedy the flaw, we would propose a coherence measure based on the convex roof extended method \cite{du2015coherence} for mixed states, then its coherence measure based on the revised transport cost is 
\begin{align}
T(\rho)=\min_{\{p_i,\ket{\phi_i}\}}\sum_i p_i\tilde{T}(\ket{\phi_i}),\label{cm}
\end{align}
where the minimum takes over all the decompositions of $\r=\sum_i p_i\ket{\phi_i}\bra{\phi_i}.$
Based on Theorem \ref{ps} and the result in \cite{du2015coherence}, we have $T(\r)$ satisfies (B1), (B3),(B4). And (B3)-(B4) can lead to (B2). Hence $T(\cdot)$ is a proper coherence measure.

\indent At last, we show the relationship between the measure and the  quantum speed limit. Here we consider a pair of pure states $\ket{\psi}$ and $\ket{\phi}$ with $\bra{\phi}\psi\rangle\in\mathbb{R}$, this is always possible, as $\ket{\psi}$ and $e^{-i\theta}\ket{\psi}$ are equivalent. Recently, Rudnicki studied the relationship between quantum speed limits and quantum entanglement measures \cite{rudnicki2021quantum}.  There the author showed that when the Hamiltonian $H^{'}$ is time-independent, 
\begin{align}
H^{'}=-i\hbar\omega(\ket{\psi}\bra{\bar{\psi}}-\ket{\bar{\psi}}\bra{\psi}),
\end{align}
here $\ket{\bar{\psi}}=\frac{\ket{\phi}-\bra{\psi}\phi\rangle\ket{\psi}}{\sqrt{1-|\bra{\phi}\psi\rangle|^2}},$
which is orthogonal to $\ket{\psi}.$ the bound (\ref{qst}) is saturated \cite{rudnicki2021quantum}. And when we denote $\tau(\psi,\phi)=\triangle t$, 
\begin{align}
\tau(\psi,\phi)=\frac{\arccos(\bra{\phi}\psi\rangle)}{\omega}.\label{lt}
\end{align}

Here we can pose a question, for a given pure state $\ket{\psi}$, how much time $\tau(\psi)$ does a unitary evolution cost by turning $\ket{\psi}$ into an incoherent state? That is,
\begin{align*}
\tau(\ket{\psi})=\frac{1}{\omega}\min_{\phi\in \mathcal{I}_{\mathcal{E}}}\arccos(|\bra{\psi}\phi\rangle|).
\end{align*} 
By the above analysis, the bound $(\ref{qst})$ is saturated when taking $H=H^{'}.$ Next assume $\ket{\psi}=\sum_{i=0}^{d-1}{\lambda_i}\ket{i}$, $|\lambda_0|\ge|\lambda_1|\ge\cdots\ge|\lambda_{d-1}|\ge0$ and $\sum_{i=0}^{d-1}|\lambda_i|^2=1,$ 
\begin{align*}
\min_{\phi\in \mathcal{I}_{\mathcal{E}}}\arccos(|\bra{\psi}\phi\rangle|)=&\arccos(\max_{\phi\in\mathcal{I}_{\mathcal{E}}}\bra{\psi}\phi\rangle)\nonumber\\
=&\arccos(|\lambda_0|)\nonumber\\
=&\arcsin(\sqrt{1-|\lambda_0|^2}),
\end{align*}

that is,
\begin{align}
\tau(\ket{\psi})=\frac{1}{\omega}\arcsin(\sqrt{1-|\lambda_0|^2}).\label{cq}
\end{align}
At last, by combing the formula (\ref{cq}) and Theorem \ref{ps}, we obtain the following theorem.
\begin{Theorem}
	Assume $\ket{\psi}$ is a pure state, then the minimal time $\tau(\ket{\psi})$ required to transform $\ket{\psi}$ into an incoherent state under a unitary evolution is 
	\begin{align}
	\tau(\ket{\psi})=\frac{1}{\omega}\arcsin(\sqrt{2T(\ket{\psi})}).
	\end{align}
\end{Theorem}

When $\ket{\psi}$ is incoherent, $T(\ket{\psi})=0,$ $\tau(\ket{\psi})=0,$ otherwise, due to the faithfulness of $T(\cdot)$, $\tau(\ket{\psi})$ cannot be 0. Furthermore, as $T(\cdot)$ is monotone under the incoherent operations for pure states, $\tau(\ket{\psi_1})\ge \tau(\ket{\psi_2})$ when $\ket{\psi_1}$ can be turned into $\ket{\psi_2}$ under incoherent operations.
  
\section{Conclusion}
\indent In this article, we studied a coherence quantifier based on the quantum transort cost. First we showed the quantity satisfies (B1), (B2), (B4), and it is faithful and subadditivity for product states. Based on the analytical solutions of the quantifier for pure states, we proposed an example showing that $\tilde{T}(\cdot)$ does not satisfies (B3). To remedy the flaw, we proposed a coherence measure based on the convex roof extended method for mixed states. At last, we obtained a close relation between the coherence measure and the minimal time necessary to evolve to an incoherent state for a pure state. Due to the importance of the study on the quantum coherence , our results can provide a reference for future work on the study of quantum coherence.

\section{Acknowledgement}
X. S. was supported by the Fundamental Research Funds for the Central Universities (Grant No. ZY2306), Funds of College of Information Science and Technology, Beijing University of Chemical Technology (Grant No. 0104/11170044115), and the National Natural Science Foundation of China (Grant No. 12301580).

\bibliographystyle{IEEEtran}
\bibliography{ref}
\section{Appendix}\label{app}
Assume $\rho_A$ and $\sigma_B$ are two quantum states in $\mathcal{D}(\mathcal{H}_A)$ and $\mathcal{D}(\mathcal{H}_B)$, respectively, then we can write the revised definition as follows.  
\begin{align}
T_s(\rho_A,\sigma_B)=&\sup\hspace{3mm} tr(\rho_AH_1+\sigma_B H_2) \label{tsa}\\
\textit{s. t.}\hspace{3mm}&  P_s-H_1\otimes I-I\otimes H_2\ge 0\nonumber\\
&  P_{as}-H_1\otimes I-I\otimes H_2\ge 0\nonumber\\
&  H_1\in \mathbb{H}(\mathcal{H}_A), H_2\in \mathbb{H}(\mathcal{H}_B),\nonumber
\end{align}
\begin{proof}
	Here the Lagrangian can be written as
	\begin{align*}
	tr[&X_{AB}P_s(d)+Y_{AB}P_{as}(d)]\nonumber\\
	+&tr[(\rho-X_A-Y_A)H_1]+tr[(\sigma-X_B-Y_B)H_2]\nonumber\\
	=tr[&X_{AB}(P_s(d)-H_1\otimes I-I\otimes H_2)]\nonumber\\
	+&tr[Y_{AB}(P_{as}-H_1\otimes I-I\otimes H_2)]\nonumber\\+&tr[\rho H_1+\sigma H_2],
	\end{align*} 
	hence the dual SDP is 
	\begin{align*}
	\sup \hspace{3mm}&tr[\rho H_1+\sigma H_2]\nonumber\\
	\textit{s. t.} \hspace{3mm}& P_s-H_1\otimes I-I\otimes H_2\ge 0\\
	&P_{as}-H_1\otimes I-I\otimes H_2\ge 0.
	\end{align*}
	
	Then we show strong duality holds for the semidefinite program $(\ref{s1}).$ Note that $(X_{AB}=0, Y_{AB}=\rho\otimes\sigma)$ is a feasible solution to the primal program (\ref{tsa}). For the dual program $(\ref{s1})$, when $H_1=H_2=-I,$ then $P_s-H_1\otimes I-I\otimes H_2=P_s+2I_A\otimes I_B,$ as for any $\ket{v}\in \mathcal{H}_{A}\otimes\mathcal{H}_B,$ $\bra{v}(P_s+2I_A\otimes I_B)\ket{v}=2+\bra{v}P_s\ket{v}\ge 2>0,$ that is, $P_s-H_1\otimes I-I\otimes H_2\ge0,$ similarly, $P_{as}-H_1\otimes I-I\otimes H_2\ge 0.$ Since a strictly feasible solution
	exists to the dual program and the primal feasible set is non-empty, the Slater’s conditions are satisfied, hence, we finish the proof  \cite{john13}.
\end{proof}

\begin{Theorem} \label{tcon}
	Assume $\rho=\sum_{i=1}^d p_i\rho_i,$ then
	\begin{align*}
	\tilde{T}(\rho)\le \sum_{i=1}^d p_i\tilde{T}(\rho_i).
	\end{align*} 
\end{Theorem}
\begin{proof}
	Assume $\sigma_i$ and $\sigma$ are the optimal incoherent states for $\r_i$ and $\rho$ in terms of $\tilde{T}(\cdot)$, respectively, $(X_{AB}^{(i)},Y_{AB}^{(i)})$ are the optimal for $T_s(\rho_i,\sigma_i)$ in terms of (\ref{s1}), 
	\begin{align}
\sum_ip_i\tilde{T}(\rho_i)=&\sum_i p_itr(X_{AB}^{(i)}P_s(d)+Y_{AB}^{(i)}P_{as}(d))\nonumber\\
\ge& T_s(\sum_ip_i\rho_i,\sum_i p_i\sigma_i)\nonumber\\
\ge&\tilde{T}(\rho),
	\end{align} 
	here the first inequality is due to the properties of $(X_{AB}^{(i)},Y_{AB}^{(i)})$ and the formula (\ref{s1}). Specifially,
	\begin{align*}
	&X_A^{(i)}+Y_A^{(i)}=\rho_i,\hspace{3mm}X_B^{(i)}+Y_B^{(i)}=\s_i\\
	\Rightarrow&\sum_i p_i(X_A^{(i)}+Y_A^{(i)})=\sum_i p_i\rho_i,\\
	&\sum_ip_i(X_B^{(i)}+Y_B^{(i)})=\sum_ip_i\s_i,
	\end{align*}
	then $(\sum_i p_iX_{AB}^{(i)},\sum_ip_i Y_{AB}^{(i)})$ satifisty the properties of (\ref{s1}). The second inequality is due to the definition of $\tilde{T}(\cdot)$.
\end{proof}

\begin{Theorem}\label{subpro}
	Assume $\rho$ and $\sigma$ are two quantum states on the system $\mathcal{H},$ then
	\begin{align*}
	\tilde{T}(\rho\otimes\sigma)\le \tilde{T}(\rho)+\tilde{T}(\sigma).
	\end{align*}
\end{Theorem}
\begin{proof}
	Let $\rho$ and $\sigma$ be two density matrices on the system $\mathcal{H},$ then we denote $\delta_1$ and $\delta_2$ be two optimal incoherent states for $\r$ and $\s$ in terms of $\tilde{T}$, respectively, that is, 
	\begin{align}
	\tilde{T}(\rho)=&\frac{1}{2}\min_{X_{AB},Y_{AB}}[tr(X_{AB}P_s(d)+Y_{AB}P_a(d))]\nonumber\\
	\emph{s.t.}\hspace{5mm}& X_A+Y_A=\rho, \hspace{3mm} X_B+Y_B=\delta_1\nonumber\\
	&X_{AB},Y_{AB}\ge 0.\label{xy}\\
	\tilde{T}(\sigma)=&\frac{1}{2}\min_{M_{AB},N_{AB}}[tr(M_{AB}P_s(d)+N_{AB}P_a(d))]\nonumber\\
	\emph{s.t.}\hspace{5mm}& M_A+N_A=\sigma, \hspace{3mm} M_B+N_B=\delta_2.\nonumber\\
	&M_{AB},N_{AB}\ge 0.\label{mn}
	\end{align}
	Here we denote $(X_{AB},Y_{AB})$ and $(M_{AB},N_{AB})$ are the optimal couples for $(\r,\delta_1)$ and $(\s,\delta_2),$ respectively. Next we compute 
	\begin{widetext}
		\begin{align}
		&\frac{1}{2}[\mathrm{tr}[M+X+N+Y+(M+X-N-Y)S]-\mathrm{tr}[(X+Y)\otimes(M+N)+(M\otimes X-N\otimes X-M\otimes Y+N\otimes Y)S\otimes S]]\nonumber\\
		=&\frac{1}{2}[1+\mathrm{tr}(M-N)S+\mathrm{tr}(X-Y)S+tr(M-N)S\otimes(X-Y)S]\nonumber\\
		=&\frac{1}{2}[(1+\mathrm{tr}(M-N)S)(1+\mathrm{tr}(X-Y)S)]\ge0\label{sub}
		\end{align}
	\end{widetext}
	As $\delta_1$ and $\delta_2$ are incoherent states, so it is with $\delta_1\otimes\delta_2$. Next as $(X_{AB},Y_{AB})$ and $(M_{AB},N_{AB})$ satisfy (\ref{xy}) and (\ref{mn}), respectively, we have 
	\begin{align*}
	(X_{A}+Y_{A})\otimes(M_A+N_A)=&\rho\otimes\s,\\
	(X_B+Y_B)\otimes(M_B+N_B)=&\delta_1\otimes\delta_2,
	\end{align*}
	hence,
	\begin{align}
	&\tilde{T}(\rho\otimes\delta)\nonumber\\
	\le& T_s(\rho\otimes\sigma,\delta_1\otimes\delta_2), \nonumber\\
	\le& \frac{1}{2}\mathrm{tr}[(X+Y)\otimes(M+N)\nonumber\\
	+&(M\otimes X-N\otimes X-M\otimes Y+N\otimes Y)S\otimes S]\label{mnxy},
	\end{align}
	then by combing (\ref{sub}) and (\ref{mnxy}), we have
	\begin{align}
	\tilde{T}(\r)+\tilde{T}(\s)\ge \tilde{T}(\r\otimes\s).
	\end{align} 
\end{proof}
\begin{Lemma}\label{rank}
	Assume $X_{AB}$ is a bipartite substate on $\mathcal{H}_d\otimes\mathcal{H}_d$, $i.$ $e.$ $X_{AB}\ge0,$ and $\mathrm{tr}X_{AB}\le 1,$ if $\rank(X_B)=1,$ or $\rank(X_A)=1,$ then $X_{AB}=X_A\otimes X_B.$
\end{Lemma}
\begin{proof}
	When $\rank(X_B)=1,$ then $X_B=m\ket{\phi}\bra{\phi},$ $m\in(0,1].$ Next let 
	\begin{align}
	X_{AB}=\begin{pmatrix}
	X_{11}&X_{12}&\cdots&X_{1d}\\
	X_{21}&X_{22}&\cdots&X_{2d}\\
	\vdots&\vdots&\ddots&\vdots\\
	X_{d1}&X_{d2}&\cdots&X_{dd}
	\end{pmatrix},
	\end{align}
	here $X_{ij}$ are the block matrices of $X_{AB},$ $i,j=1,2,\cdots d,$ $X_{kk}\ge0,$ $k=1,2,\cdots,d,$ $X_B=X_{11}+X_{22}+\cdots+X_{dd}=m\ket{\phi}\bra{\phi},$ that is, $X_{kk}=l_{kk}\ket{\phi}\bra{\phi},$ $\sum_k l_{kk}=m.$ Next as $X_{AB}$ is semidefinite positive, then 
	\begin{align*}
	\begin{pmatrix}
	X_{ii}& X_{ij}\\
	X_{ji}& X_{kk}
	\end{pmatrix}\ge 0,
	\end{align*}
	hence $X_{ij}=l_{ij}\ket{\phi}\bra{\phi},$ that is,
	\begin{align*}
	X_{AB}=\begin{pmatrix}
	l_{11} \phi&l_{12}\phi&\cdots&l_{1d}\phi\\
		l_{21}\phi&l_{22}\phi&\cdots&l_{2d}\phi\\
	\vdots&\vdots&\ddots&\vdots\\
	l_{d1}\phi&l_{d2}\phi&\cdots&l_{dd}\phi
	\end{pmatrix}=M\otimes \phi.
	\end{align*}
	If $\rank(X_A)=1,$ let $Y_{AB}=S X_{AB} S^{\dagger},$ then $\rank(Y_{B})=1,$ due to the above proof, $Y_{AB}$ can be written as the product form. Hence, $X_{AB}$ can be written as the product form.
 \end{proof}

\end{document}